\documentclass[11pt]{article}

\usepackage[T1]{fontenc}
\usepackage{lmodern}
\usepackage[margin=1in]{geometry}
\usepackage{microtype}
\usepackage{amsmath,amssymb,amsthm,mathtools}
\usepackage{enumitem}
\usepackage{xurl}
\usepackage{xcolor}
\usepackage[
  backend=biber,
  style=numeric,
  sorting=none,
  giveninits=true,
  maxbibnames=99,
  doi=true,
  isbn=false,
  url=true,
  eprint=true
]{biblatex}
\usepackage[hidelinks]{hyperref}
\usepackage[nameinlink,noabbrev]{cleveref}

\setlist{itemsep=0.35em,topsep=0.45em}
\newtheorem{theorem}{Theorem}[section]
\newtheorem{lemma}[theorem]{Lemma}
\newtheorem{proposition}[theorem]{Proposition}

\newtheorem{conjecture}[theorem]{Conjecture}
\theoremstyle{definition}

\theoremstyle{remark}
\newtheorem{remark}[theorem]{Remark}

\newcommand{\F}{\mathbb{F}}
\newcommand{\Q}{\mathbb{Q}}
\newcommand{\R}{\mathbb{R}}
\newcommand{\Z}{\mathbb{Z}}
\newcommand{\OO}{\mathcal{O}}
\newcommand{\softO}{\widetilde{O}}
\newcommand{\polylog}{\operatorname{polylog}}
\DeclareMathOperator{\poly}{poly}
\newcommand{\Norm}{\operatorname{N}}

\title{Quantum Algorithms for Modular Factorials}
\author{
    Yann Tal\\
    \texttt{yanntal@mail.tau.ac.il}
}
\date{July 2026}
\hypersetup{
  pdftitle={Quantum Algorithms for Modular Factorials},
  pdfauthor={Yann Tal}
}

\begin{document}
\maketitle

\begin{abstract}
We give a bounded-error quantum algorithm that, given a prime \(p\), a divisor \(q\mid(p-1)\), and an integer \(0<n<p\), computes
\(n!\bmod p\) in expected time
\[
    \softO\!\left(q^c+\sqrt{p/q}\right)
\]
for some absolute constant \(c\ge1\). When \(p-1\) has a  divisor
of size \(q\approx p^{1/(2c+1)}\), this gives the exponent
\(c/(2c+1)<1/2\). To our knowledge, this is the first algorithm to break
the exponent \(1/2\) barrier for modular factorials under such a
divisor promise.

The main technical ingredient is a quantum algorithm that reconstructs the
relevant Jacobi sum exactly in compact algebraic form, with polynomial
dependence on \(q\) and \(\log p\). We further extend the same asymptotic
bound to the computation of \(n!\bmod p^2\), uniformly over
\(0\le n<p^2\). At \(n=p-1\), this determines the Wilson quotient
\[
    \frac{(p-1)!+1}{p}
    \pmod p.
\]

We conjecture that the condition \(q\mid(p-1)\) is a technical limitation of the present method
rather than an inherent obstruction, and a uniform quantum algorithm exists for all primes.
\end{abstract}

\clearpage
\tableofcontents
\clearpage

\section{Introduction}

\subsection{Classical Algorithms and Main Results}

Let \(p\) be prime and let \(0<n<p\). We further assume that a divisor \(q\)
of \(p-1\), of suitable intermediate size, exists and is provided to us. Our
goal is to compute
\[
    n! = 1\cdot 2\cdots n \pmod p.
\]
However, direct multiplication requires \(n-1\) operations in \(\F_p\). The
natural question studied in this work is whether this product can be evaluated
without paying for essentially every individual factor.\footnote{One may
instead multiply only the primes up to \(n\), raised to their respective
multiplicities in \(n!\), but this saves at most \(\polylog\) factors.}

The fastest known classical algorithm for this problem is due to Bostan,
Gaudry, and Schost \cite{BGS}; however, it improves only by \(\polylog\)
factors over the simpler algorithm known as the baby-step/giant-step product
algorithm, which computes \(n!\bmod p\) using
\[
    \softO(\sqrt n)
\]
field operations. To illustrate the idea, suppose for simplicity that \(n\)
is a perfect square, and define
\[
    f(x)=(x+1)(x+2)\cdots(x+\sqrt n).
\]
Then
\[
    f(0)=1\cdot2\cdots\sqrt n,
\]
\[
    f(\sqrt n)=(\sqrt n+1)\cdots 2\sqrt n,
\]
and similarly for the remaining consecutive blocks. Therefore,
\[
    \prod_{i=0}^{\sqrt n-1} f(i\sqrt n)
    \equiv n!\pmod p.
\]

The polynomial \(f\) has degree \(\sqrt n\) and is evaluated at \(\sqrt n\)
points. A product tree constructs \(f\), and fast multipoint evaluation
computes all these values together using
\[
    \softO(\sqrt n)
\]
field operations. This procedure also requires \(\softO(\sqrt n)\) field
elements of workspace. One may reduce the space usage by working with shorter
blocks, but this increases the number of evaluations and hence the running
time. Wilson's theorem also allows one to apply the same method to the
complementary product:
\[
    n!
    \equiv
    -\left(\prod_{j=n+1}^{p-1}j\right)^{-1}
    \pmod p.
\]
Thus the elementary classical bound is
\[
    \softO\!\left(1+\sqrt{\min\{n,p-1-n\}}\right),
\]
which remains \(\softO(\sqrt p)\) in the worst case.

Our main contribution is the following.

\begin{theorem}\label{thm:main}
There exists a bounded-error quantum algorithm that, given a prime \(p\) and
natural numbers \(q,n\), where \(q\mid(p-1)\) and \(0<n<p\), computes
\[
    n!\pmod p
\]
in expected time
\[
    \softO\left(q^c+\sqrt{\frac{p}{q}}\right),
\]
for some absolute constant \(c\ge 1\).
\end{theorem}

The two terms in the running time are balanced when
\[
    q\approx p^{\frac{1}{2c+1}}.
\]
Whenever \(p-1\) has a divisor of this size, the worst-case running time
becomes
\[
    \softO\left(p^{\frac{c}{2c+1}}\right),
\]
and therefore achieves the exponent
\[
    \frac{c}{2c+1}
    =
    \frac{1}{2}-\frac{1}{4c+2}
    <
    \frac{1}{2},
\]
breaking the square-root barrier.

Prime-square moduli are also of independent arithmetic interest. For a prime
\(p\), the Wilson quotient is
\[
    W_p=\frac{(p-1)!+1}{p}.
\]
Consequently, computing \((p-1)!\bmod p^2\) determines \(W_p\bmod p\), and
\(p\) is a Wilson prime precisely when
\[
    (p-1)!\equiv -1\pmod{p^2}.
\]
Costa, Gerbicz, and Harvey \cite{CostaGerbiczHarvey} used this connection in
their search for Wilson primes, giving an algorithm that computes these
residues in average polynomial time per prime when many primes are treated
simultaneously. Still, it was not known how to break the square-root barrier, when one is interested in a worst-case algorithm on a single input. Our second main
theorem, treats modular factorials modulo \(p^2\) for a single
input.

Our second main contribution is a refinement of \Cref{thm:main} to
prime-square moduli.

\begin{theorem}\label{thm:prime-square}
Let \(c\) be the absolute constant from \Cref{thm:main}. There exists a
bounded-error quantum algorithm such that, given a prime \(p\) and natural
numbers \(q,n\), where \(q\mid(p-1)\) and \(0\le n<p^2\), the algorithm
computes
\[
    n!\pmod{p^2}
\]
in expected time
\[
    \softO\left(q^c+\sqrt{\frac{p}{q}}\right),
\]
\end{theorem}

Breaking the square-root bound again.

\subsection{The Technique}

The key fact underlying the algorithm in the paper, is that for 
\[
    K=\frac{p-1}{q}.
\]
there is a quantum algorithm computing 
$(aK)! \pmod p$, for 
any multiple $1 \le a < q$,
such that the quantum algorithm has a non-trivial running time, namely, it runs in time $\softO(q^{c_0})$ for some constant $c_0$. If $q$ is small enough (e.g., if $q<n^{1/2c_0}$), this running time is smaller than the square-root barrier we discussed above. Once we have such an algorithm, we can compute $n! \pmod p$ by finding the closest multiple of $K$ below $n$, and completing the remaining terms $(aK+1) \cdot \ldots \cdot n$ by the baby-step giant-step classical algorithm.

The main question is how to compute factorials of multiples of $K$.
Here, we use a specialization of Young's \(p\)-adic congruences relating Jacobi sums and multinomial coefficients \cite[Theorem~2.2]{Young1995}, giving a closed form for \((aK)! \pmod p\).\footnote{A slightly more complicated closed form modulo \(p^2\) is given in Proposition~\ref{prop:lifted-central}.}
It says:
for every \(1\le a<q\), 
\begin{equation}\label{eq:central}
    (aK)!
    \equiv
    (-1)^{a-1}
    \rho_{\mathfrak p}\!\left(J_a(\chi)\right)
    (K!)^a
    \pmod p.
\end{equation}
The notation in this identity is introduced in \Cref{sec:preliminaries},
and we recommend the reader to overlook it in a first reading. The main thing to take from this formula
is that if one can compute $K! \pmod p$ and a Jacobi sum modulo a prime ideal $\mathfrak{p}$ related to $p$ and $q$, then one can also compute $(aK)!$.

We give a quantum algorithm for computing the Jacobi sum modulo the ideal $\mathfrak{p}$ in time $\softO(q^{c_0})$. 
Computing $K! \pmod p$ can be done in time $\softO(K^{1/2})$ using the baby-step giant-step 
algorithm.  Altogether, the running time $O(q^{c_0}+\sqrt{\frac{p}{q}})$ as promised. So we now turn to discuss the idea behind the Jacobi part.

There is a quantum algorithm for approximating $J_a(\chi)$ as a complex number. Van Dam and Seroussi \cite{VanDamSeroussi} gave a quantum algorithm for estimating the phase of a Gauss sum, and with that one can estimate the phase of the Jacobi sum $J_a(\chi)$. The magnitude of the Jacobi sum $J_a(\chi)$ is known, and together these give a numerical approximation to $J_a(\chi)$. The main problem that we face is that we need the Jacobi sum modulo the prime ideal $\mathfrak p$, which such a numerical approximation does not directly provide.

\paragraph{High-level idea.}
The main difficulty is that the Jacobi sum is an algebraic integer in the
cyclotomic field $K=\mathbb{Q}(\zeta_q)$, while the algebraic generators
produced by the number-field algorithms used below may have enormous height and
therefore prohibitively large expanded representations. Thus, even though we
only seek the residue of the Jacobi sum modulo the prime ideal
$\mathfrak p$, it is not clear how to recover this residue from such
generators without first expanding them.

Our approach is to avoid constructing the Jacobi sum explicitly. Instead, we
compute a \emph{compact representation} of it. More precisely, we first compute
the principal ideal generated by the Jacobi sum, and then use quantum algorithms
for principal ideals and unit groups to recover a compact algebraic description
of a generator of this ideal. This representation has size polynomial in
$q$ and $\log p$, yet contains sufficient information to identify the
Jacobi sum and evaluate it modulo $\mathfrak p$. Consequently, we obtain the
desired residue of the Jacobi sum modulo $\mathfrak p$ without ever expanding
its compact algebraic representation in the number field. We perform the
multiplication represented by the compact power product only after reducing
modulo $\mathfrak p$, when the computation takes place in the much smaller
residue field.

Thus, the central idea is to replace the expansion of large intermediate
algebraic numbers by a succinct algebraic representation that is sufficient
for the desired modular reduction.

\paragraph{The algorithmic realization of the above program.} We now need to discuss some of the number-theoretic details involved. The Jacobi sum \(J_a(\chi)\) may be viewed as a complex number, but for our purposes its more useful representation is as an algebraic integer \(\sum_{i=0}^{\varphi(q)-1}c_i\zeta_q^i\in\Z[\zeta_q]\), where the \(c_i\) are integers and, since \(\chi\) is a character over \(\F_p\) of order \(q\), \(\zeta_q\) is a primitive \(q\)-th root of unity. For our application, what we need is not \(J_a(\chi)\) itself, but rather \(J_a(\chi)\pmod{\mathfrak p}\), where \(\mathfrak p\) is a prime ideal of \(\Z[\zeta_q]\).

Our algorithm has three main steps:
\begin{enumerate}
\item We look at the ideal generated by \(J_a(\chi)=\sum_{i=0}^{\varphi(q)-1}c_i\zeta_q^i\) in \(\Z[\zeta_q]\). While we do not know the coefficients \(c_i\), the \emph{ideal} generated by \(J_a(\chi)\) is easier to determine. Stickelberger's theorem tells us explicitly how it decomposes into a product of prime ideals. A main step in the algorithm is to find an element \(\gamma\in\Z[\zeta_q]\) that generates the same ideal. For that we use the Biasse--Song PIP algorithm, detailed in \Cref{thm:pip-black-box}.
\item We now use \(\gamma\) to find another element \(\alpha\in\Z[\zeta_q]\) that generates the same ideal and has the additional property that it differs from \(J_a(\chi)\) only by multiplication by a root of unity. For that we use the Biasse--Song \(S\)-unit algorithm twice, detailed in \Cref{thm:s-unit-black-box}.
\item Finally, we are left with finding this root of unity, and for that we use the quantum Gauss-sum phase-estimation algorithm of Van Dam and Seroussi \cite{VanDamSeroussi}, stated in \Cref{thm:gauss-phase-black-box}. Together, we recover \(J_a(\chi)\) in an exact compact representation that allows us to reduce it modulo \(\mathfrak p\) efficiently.
\end{enumerate}
The beauty is that the above steps use compact representations throughout, allowing us finally to perform the reduction modulo \(\mathfrak p\). Another encouraging feature is that, although many quantum algorithms in computational number theory rely on unproven assumptions such as the Generalized Riemann Hypothesis, all the routines used here are unconditional.

%The reduction underlying our algorithm can be summarized as follows:
%\[
%    J_a(\chi)   \xrightarrow{\text{reduction modulo }\mathfrak p}
%\rho_{\mathfrak p}\!\left(J_a(\chi)\right)
%\xrightarrow{\text{central congruence}}
%(aK)!\pmod p
%\xrightarrow{\text{interval product}}
%n!\pmod p.
%\]

\subsubsection{Exact Reconstruction of the Jacobi Sum}

We think that it is also worthwhile to state what we have obtained for the reconstructing Jacobi-sum exactly in compact algebraic form.
For fixed character order, Buhler and Koblitz gave an LLL-based classical
algorithm that computes the corresponding Jacobi sums in time
\(O(\log^3 p)\) \cite{BuhlerKoblitz1998}. Their analysis treats the character
order as fixed. The regime needed here is different: \(q\) grows with the
input and is the parameter that yields the factorial speedup. The following
theorem gives polynomial dependence on both \(q\) and \(\log p\), and more
strongly reconstructs the Jacobi sum exactly in compact algebraic form.

\begin{theorem}[Exact Jacobi-sum reconstruction]
\label{thm:jacobi-computation}
There exists an absolute constant \(c\ge 1\) and a bounded-error quantum
algorithm with the following property. Let \(p\) be prime, let \(q\) divide
\(p-1\), and let \(1\le a<q\). Given the character \(\chi\) and the prime
ideal \(\mathfrak p\) appearing in \eqref{eq:central}, the algorithm outputs
an exact compact representation of
\[
    J_a(\chi)\in\OO_L
\]
of total description length \(q^{O(1)}\polylog(p)\), in expected time
\[
    \softO(q^c).
\]
Consequently, for each \(h\in\{1,2\}\), it also computes
\[
    \rho_{\mathfrak p^h}\!\left(J_a(\chi)\right)
    \in
    \Z/p^h\Z
\]
within the same asymptotic bound.
\end{theorem}

Taking \(h=1\) in \Cref{thm:jacobi-computation}, and combining the
result with the central congruence and the two classical interval products,
gives the running time
\[
    \softO\left(q^c+\sqrt{\frac{p}{q}}\right)
\]
claimed in \Cref{thm:main}.

\subsection{How restrictive is the assumption on \(q\)?}

We first ask whether a suitable divisor \(q\) can be found efficiently. At
first sight, this appears to introduce another difficult search problem: one
must first factor \(p-1\), and then select a divisor of approximately the right
size. The first step can be carried out in quantum polynomial time using
Shor's algorithm \cite{Shor}. The second is less obviously easy: even given
the complete prime factorization of \(p-1\), choosing which prime powers to
include resembles a multiplicative version of Subset Sum.

Fortunately, an exact solution is unnecessary. It is enough to find a divisor
whose size is within a small multiplicative factor of the best available
choice. This can be done efficiently by a standard trimming argument.

\begin{proposition}
\label{prop:approximate-divisor-selection}
Given the complete prime factorization of \(p-1\), a target \(T\leq p-1\),
and \(0<\varepsilon\leq 1\), one can deterministically find a divisor
\(q\mid p-1\) such that
\[
    \frac{1}{1+\varepsilon}
    \max\{d:d\mid p-1,\ d\leq T\}
    \leq q\leq T
\]
in time polynomial in \(\log p\) and \(1/\varepsilon\).
\end{proposition}

\begin{proof}
Write the factorization
\[
    p-1=r_1r_2\cdots r_m,
\]
where the prime factors are repeated according to their multiplicities. Since
every \(r_i\geq 2\), we have \(m\leq \log_2 p\).

The idea is to process the factors one at a time while retaining only a sparse
list of representative divisors. Set
\[
    \delta=\frac{\varepsilon}{2m}
\]
and begin with \(L_0=\{1\}\). After constructing \(L_{i-1}\), form
\[
    L_{i-1}\cup
    \{r_i x:x\in L_{i-1},\ r_i x\leq T\}.
\]
Sort the resulting list. Retain its smallest element, and thereafter retain an
element only when it is larger than \(1+\delta\) times the previously retained
element. Denote the trimmed list by \(L_i\).

This is the standard trimmed-list approximation for Subset Sum
\cite{IbarraKim1975}, applied multiplicatively. We claim that, after the first
\(i\) factors have been processed, every divisor \(x\leq T\) of
\(r_1\cdots r_i\) has a representative \(y\in L_i\) satisfying
\[
    y\leq x\leq (1+\delta)^i y.
\]
Indeed, before trimming, either \(x\) itself or the product of \(r_i\) with a
representative from \(L_{i-1}\) occurs in the new list. Trimming loses at most
one additional factor of \(1+\delta\), proving the claim by induction.

Consequently, the largest element \(q\) of \(L_m\) satisfies
\[
    q\geq
    \frac{1}{(1+\delta)^m}
    \max\{d:d\mid p-1,\ d\leq T\}.
\]
Since
\[
    (1+\delta)^m
    \leq \exp(\varepsilon/2)
    \leq 1+\varepsilon,
\]
the required approximation follows.

It remains to bound the size of the lists. Consecutive elements of every
\(L_i\) differ by a factor greater than \(1+\delta\), and all elements lie
between \(1\) and \(T\). Therefore,
\[
    |L_i|
    \leq 1+\log_{1+\delta}T
    =
    O\!\left(\frac{(\log p)^2}{\varepsilon}\right).
\]
The lists can therefore be constructed, merged, and trimmed using polynomially
many arithmetic operations on \(O(\log p)\)-bit integers.
\end{proof}

To apply the proposition, take
\[
    T=p^{1/(2c+1)},
\]
the scale at which \(q^c\) and \(\sqrt{p/q}\) are balanced. Applying
Proposition~\ref{prop:approximate-divisor-selection} with target \(T\)
approximates the largest divisor of \(p-1\) below this scale. Applying it with
target \((p-1)/T\), and then taking the complementary divisor, approximates the
smallest divisor above it. For every fixed \(\varepsilon>0\), choosing the
better of these two divisors gives, in time polynomial in \(\log p\), a
running-time bound within a constant factor of the best obtainable from any
divisor of \(p-1\).

Since Shor's algorithm factors \(p-1\) in quantum polynomial time, \(q\) need
not be supplied as part of the input in the quantum setting. The remaining restriction is that
\(p-1\) must possess a divisor of suitable intermediate size for the running
time to beat the square-root bound.

We next ask how often such a divisor exists. Koukoulopoulos \cite{Koukoulopoulos2010} proved that, for
every fixed \(0<\alpha<\beta<1\) in the relevant range, a positive proportion
of primes \(p\) have some divisor \(q\mid p-1\) satisfying
\[
    p^\alpha<q\leq p^\beta.
\]
Thus divisors of the intermediate sizes required by our algorithm occur for a
positive proportion of primes, although it is not clear how large this
proportion is.

\subsection{A remark about Factorials and Integer Factoring}

All known fast classical algorithms for computing factorials are
deterministic. Strassen \cite{Strassen1976} observed that the
baby-step/giant-step algorithm for computing modular factorials can also be
used to solve the integer factoring problem. His method works as follows.
Assume for simplicity that \(N\) is semiprime and factors as
\[
    N=p_1p_2,
\]
where \(p_1<p_2\). Then
\[
    \gcd\!\left(\lfloor\sqrt N\rfloor!\bmod N,N\right)
    =
    p_1.
\]
To see this, notice that
\[
    p_1\le\lfloor\sqrt N\rfloor<p_2.
\]
Applying the baby-step/giant-step procedure defined above allows one to
compute
\[
    \lfloor\sqrt N\rfloor!\pmod N
\]
deterministically in
\[
    \softO(N^{1/4})
\]
time, and thus allows one to factor \(N\) in similar time.

The connection between factorial computation and integer factorization
extends beyond Strassen's algorithm. Lipton \cite{Lipton2010} observed that if
\(n!\) admitted straight-line programs of length \(\polylog(n)\), then integer
factorization would have polynomial-size circuits; if these programs could be
constructed uniformly, the same argument would yield a polynomial-time
factoring algorithm. Indeed, one may compute \(y!\bmod N\) for several values
of \(y\) and use binary search on
\[
    \gcd(y!,N)
\]
to recover a proper factor of \(N\).

Interestingly, Lipton closes his \emph{Factoring and Factorials} chapter with
the following question:
\[
    \textit{Would a (fast) factoring algorithm imply a fast algorithm for
    \(n!\bmod p\)?}
\]
No such implication is currently known. Integer factorization has succinctly
verifiable certificates, whereas no analogous verification procedure for
\(n!\bmod p\) is apparent. In particular, membership of the natural decision
problem in \(PH\) or \(QMA\) remains unclear.

The straight-line complexity of factorials has also been studied in the
Blum--Shub--Smale model \cite{BSS1989}. Shub and Smale
\cite{ShubSmale} considered the relaxed task of computing a nonzero multiple
\[
    m_n n!
\]
by a short straight-line program, and showed that hardness of this task would
imply
\[
    P_{\mathbb C}\neq NP_{\mathbb C}.
\]

Using ideas from Lenstra's elliptic-curve factoring algorithm
\cite{Lenstra1987}, Cheng \cite{Cheng2004} gave, under a conjecture on smooth
integers in short intervals, a randomized construction of such a program of
length
\[
    \exp\!\left(
        O\!\left(\sqrt{\log n\log\log n}\right)
    \right).
\]
This does not directly aid in computing \(n!\bmod p\), since the multiplier
\(m_n\) is not prescribed or recovered and may be divisible by \(p\).

Efficient integer factorization (that indeed is possible in $BQP$) can nevertheless assist, although modestly, in computing modular
factorials over composite moduli. Indeed, given the factorization
\[
    N=\prod_i p_i^{e_i},
\]
the Chinese remainder theorem reduces the computation of \(n!\bmod N\) to
the corresponding computations modulo the prime powers \(p_i^{e_i}\). In
particular, when every exponent satisfies \(e_i\in\{1,2\}\), the algorithms
developed here for prime and prime-square moduli may be applied separately to
the relevant components, subject to the same divisor condition on each
\(p_i-1\). The resulting residues can then be recombined modulo \(N\) using
the Chinese remainder theorem.

The extension from \(p^2\) to higher prime powers is a natural open
direction, to which we return in \Cref{sec:conclusion}.

\subsection{Organization of the Paper}

\Cref{sec:preliminaries} introduces the algebraic and computational
background used throughout the paper. It defines the cyclotomic field \(L=\Q(\zeta_q)\),
its maximal real subfield, the prime ideals above \(p\), and the reduction
maps modulo \(\mathfrak p\) and \(\mathfrak p^2\). It also introduces
multiplicative characters, Gauss sums, Jacobi sums, unit groups,
\(S\)-units, compact representations, and the classical interval-product
algorithm. Finally, it states the three external quantum algorithms used in
the Jacobi-sum reconstruction: the Biasse--Song algorithms for computing
\(S\)-unit groups and solving the principal ideal problem, and the
Van Dam--Seroussi algorithm for estimating Gauss-sum phases.

\Cref{sec:prime-modulus} presents the five-step procedure for recovering
\(J_a(\chi)\) exactly in compact algebraic form and reducing it modulo powers
of \(\mathfrak p\). \Cref{sec:correctness} proves the correctness of this
procedure and analyzes its time and workspace requirements, including the
bit sizes of its inputs, outputs, and intermediate representations. The
resulting Jacobi-sum algorithm is then combined with two classical interval
products to prove \Cref{thm:main}.

\Cref{sec:prime-square} extends the method to computation modulo \(p^2\).
The compact representation of \(J_a(\chi)\) obtained by the prime-modulus
algorithm is reduced modulo \(\mathfrak p^2\), and a central congruence
modulo \(p^2\) relates this reduction to \((aK)!\). The additional harmonic
term in this congruence is evaluated using a roots-of-unity filter, yielding
\Cref{thm:prime-square} with the same asymptotic running time.

\Cref{sec:conclusion} summarizes the limitations of the present method and
records the main open problems.

The appendices contain four calculations whose full proofs would otherwise
interrupt the algorithmic development. \Cref{app:central-congruence}
proves the central congruence modulo \(p\).
While it is a special case of Young's \(p\)-adic
congruences, it is much simpler, and we give a short proof of it for completeness.
\Cref{app:gauss-jacobi} proves the Gauss--Jacobi relation and the norm
identity for \(J_a(\chi)\). \Cref{app:stickelberger} derives the
Stickelberger factorization of the principal ideal generated by the
\(a\)-fold Jacobi sum. Finally, \Cref{app:lifted-central} proves the
central congruence modulo \(p^2\) by specializing Young's
Jacobi--multinomial congruence and applying an elementary block expansion
modulo \(p^2\).

%\ynote{Should change \(c>0\) to \(c_0\).}

\section{Preliminaries}\label{sec:preliminaries}

This section fixes the algebraic notation, computational representations, and
external quantum algorithms used throughout the paper.

\subsection{Number Fields, Ideals, and Reduction}

For a number field \(F\), the notation
\[
    [F:\Q]
\]
denotes its degree over \(\Q\), and \(\OO_F\) denotes its ring of integers.
The group of units of \(\OO_F\) is written \(\OO_F^\times\).

For a nonzero element \(\beta\in\OO_F\), the notation
\[
    (\beta)=\beta\OO_F
\]
denotes the principal ideal generated by \(\beta\). If
\(\mathfrak a\subseteq\OO_F\) is a nonzero ideal, its absolute norm is
\[
    \Norm(\mathfrak a)
    =
    |\OO_F/\mathfrak a|.
\]
For a prime ideal \(\mathfrak q\), we write
\(v_{\mathfrak q}(\beta)\) for the exponent of \(\mathfrak q\) in the
factorization of \((\beta)\).

Let
\[
    \zeta_q:=e^{2\pi i/q},
    \qquad
    L:=\Q(\zeta_q),
    \qquad
    \OO_L:=\Z[\zeta_q].
\]
Complex conjugation sends
\[
    \zeta_q\longmapsto\zeta_q^{-1}
\]
and is denoted by an overline. Its fixed field is the maximal real subfield
\[
    L^+:=\Q(\zeta_q+\zeta_q^{-1}),
\]
whose ring of integers is denoted by \(\OO_{L^+}\). We write \(\mu_q\) for
the group of \(q\)-th roots of unity and \(\mu(L)\) for the group of all
roots of unity contained in \(L\).

Throughout \Cref{sec:prime-modulus} and \Cref{sec:correctness}, let \(q\mid p-1\) and set
\[
    K=\frac{p-1}{q}.
\]
Choose a primitive root \(g\in\F_p^\times\), and let
\[
    \eta=g^{-K}\in\F_p.
\]
Then \(\eta\) has order \(q\). Define
\[
    \mathfrak p:=(p,\zeta_q-\eta)\subseteq\OO_L.
\]
Reduction modulo \(\mathfrak p\) gives the surjective ring homomorphism
\[
    \rho_{\mathfrak p}\colon\OO_L\longrightarrow\F_p,
    \qquad
    \rho_{\mathfrak p}\!\left(f(\zeta_q)\right)
    =
    f(\eta)\pmod p.
\]
Its kernel is \(\mathfrak p\), and hence
\[
    \OO_L/\mathfrak p\simeq\F_p.
\]

For
\[
    t\in(\Z/q\Z)^\times,
\]
define
\[
    \sigma_t(\zeta_q):=\zeta_q^t,
    \qquad
    \mathfrak p_t
    :=
    \sigma_t^{-1}(\mathfrak p)
    =
    (p,\zeta_q-\eta^t).
\]
Since \(q\mid p-1\), the prime \(p\) splits completely in \(L\), and the
ideals \(\mathfrak p_t\) are precisely the prime ideals of \(\OO_L\) above
\(p\).

More generally, since \(\eta\) has order \(q\), it is a root of
\(\Phi_q(X)\) modulo \(p\). As \(p\nmid q\), the polynomial
\(X^q-1\) is square-free modulo \(p\), so \(\eta\) is a simple root of
\(\Phi_q\). Hence, for every integer \(h\ge1\), it lifts uniquely to a root
\(\eta_h\) of \(\Phi_q(X)\) modulo \(p^h\) satisfying
\[
    \eta_h\equiv\eta\pmod p.
\]
Evaluation at \(\eta_h\) gives a surjective ring homomorphism
\[
    \rho_{\mathfrak p^h}\colon
    \OO_L\longrightarrow\Z/p^h\Z,
    \qquad
    \zeta_q\longmapsto\eta_h.
\]
The ideal \(\mathfrak p^h\) is contained in its kernel. Since
\[
    |\OO_L/\mathfrak p^h|
    =
    \Norm(\mathfrak p)^h
    =
    p^h,
\]
the kernel is exactly \(\mathfrak p^h\), and hence
\[
    \OO_L/\mathfrak p^h\simeq\Z/p^h\Z.
\]
For \(h=1\), this is the map \(\rho_{\mathfrak p}\) above. For \(h=2\),
we write \(\widetilde\eta:=\eta_2\). Whenever
\(x=B/C\in L\) with \(B,C\in\OO_L\) and \(\mathfrak p\nmid C\), we also
write
\[
    \rho_{\mathfrak p^h}(x)
    =
    \rho_{\mathfrak p^h}(B)
    \rho_{\mathfrak p^h}(C)^{-1}.
\]
Every \(x\in L^\times\) for which the exponent of \(\mathfrak p\) in the
fractional ideal \((x)\) is nonnegative admits such a presentation. Thus
the same notation is defined for every such \(x\); when this exponent is
zero, \(\rho_{\mathfrak p^h}(x)\) is invertible modulo \(p^h\).

\subsection{Characters, Gauss Sums, and Jacobi Sums}

Define the multiplicative character
\[
    \chi\colon\F_p^\times\longrightarrow\mu_q
\]
by
\[
    \chi(g)=\zeta_q,
\]
and extend it to \(\F_p\) by setting \(\chi(0)=0\).

For a nontrivial multiplicative character \(\psi\) of \(\F_p^\times\),
define its Gauss sum by
\[
    G(\psi)
    :=
    \sum_{x\in\F_p}
    \psi(x)\exp\left(\frac{2\pi i x}{p}\right).
\]
For \(1\le a<q\), define the \(a\)-fold Jacobi sum by
\[
    J_a(\chi)
    :=
    \sum_{\substack{x_1,\ldots,x_a\in\F_p\\
                    x_1+\cdots+x_a=1}}
    \prod_{i=1}^{a}\chi(x_i).
\]
In particular,
\[
    J_1(\chi)=1.
\]

\begin{proposition}[Gauss--Jacobi identities]
\label{prop:gauss-jacobi}
For every \(1\le a<q\),
\[
    J_a(\chi)
    =
    \frac{G(\chi)^a}{G(\chi^a)}
\]
and
\[
    J_a(\chi)\overline{J_a(\chi)}
    =
    p^{a-1}.
\]
\end{proposition}

The proof is given in \Cref{app:gauss-jacobi}.

\begin{proposition}[Stickelberger factorization]
\label{prop:stickelberger-factorization}
For every \(1\le a<q\),
\[
    \bigl(J_a(\chi)\bigr)
    =
    \prod_{t\in(\Z/q\Z)^\times}
    \mathfrak p_t^{\left\lfloor at/q\right\rfloor}.
\]
\end{proposition}

A derivation from the standard two-character Stickelberger factorization is
given in \Cref{app:stickelberger}; see also
\cite[Proposition~4(ii)]{Katre2000}.

\subsection{Unit Groups, \texorpdfstring{\(S\)}{S}-Units, and Compact Representations}

Dirichlet's unit theorem gives decompositions
\[
    \OO_L^\times
    =
    \mu(L)\times
    \langle\epsilon_1,\ldots,\epsilon_r\rangle
\]
and
\[
    \OO_{L^+}^\times
    =
    \{\pm1\}\times
    \langle e_1,\ldots,e_r\rangle,
\]
where
\[
    r=\frac{\varphi(q)}{2}-1,
\] and \(\varphi(\cdot)\) is the Euler totient function.

The units
\[
    \epsilon_1,\ldots,\epsilon_r
    \qquad\text{and}\qquad
    e_1,\ldots,e_r
\]
are fundamental-unit bases for the free parts of the two unit groups.

The relative norm from \(L\) to \(L^+\) is
\[
    \Norm_{L/L^+}(\beta)
    =
    \beta\overline\beta.
\]
In particular, the relative norm maps units of \(\OO_L\) to units of
\(\OO_{L^+}\).

Let \(F\) be a number field and let \(S\) be a finite set of prime ideals of
\(\OO_F\). The \(S\)-unit group is
\[
    \OO_{F,S}^\times
    =
    \left\{
        \beta\in F^\times:
        v_{\mathfrak q}(\beta)=0
        \text{ for every }\mathfrak q\notin S
    \right\}.
\]
When \(S=\varnothing\), this is the ordinary unit group:
\[
    \OO_{F,\varnothing}^\times=\OO_F^\times.
\]

A compact representation is a symbolic power product
\[
    \prod_{j=1}^{m}\beta_j^{z_j},
\]
where each \(\beta_j\in F^\times\) is specified exactly by rational
coordinates in the fixed integral basis of \(F\), and \(z_j\in\Z\). The size
of the representation is the total description length of this list. We use
only compact representations of polynomial size. Such a representation may
nevertheless describe an algebraic number of very large height, since the
product is not expanded.

\subsection{Computational Representations}

The number-field algorithms used below require explicit descriptions of their
input fields, rings, and ideals.

A number field is presented as
\[
    F=\Q(\theta)=\Q[X]/(f(X)),
\]
where \(f\in\Z[X]\) is monic and irreducible. Its ring of integers is
supplied by an integral basis
\[
    \OO_F
    =
    \Z\omega_1\oplus\cdots\oplus\Z\omega_{[F:\Q]}.
\]
Each \(\omega_i\) is represented as a polynomial in \(\theta\) of degree
less than \([F:\Q]\), with rational coefficients.

An ideal \(\mathfrak a\subseteq\OO_F\) is represented by a
\(\Z\)-basis relative to the chosen integral basis, usually in Hermite normal
form. When an ideal factorization is already available, we also use the
factored representation
\[
    \mathfrak a
    =
    \prod_i\mathfrak q_i^{m_i},
\]
stored as the list
\[
    \{(\mathfrak q_i,m_i)\}_i.
\]

Let \(\Phi_q(X)\in\Z[X]\) denote the \(q\)-th cyclotomic polynomial,
\[
    \Phi_q(X)
    :=
    \prod_{t\in(\Z/q\Z)^\times}
    \left(X-\zeta_q^t\right).
\]
Equivalently, \(\Phi_q(X)\) is the minimal polynomial of
\(\zeta_q\) over \(\Q\).

For the cyclotomic field, we use the presentation
\[
    L
    \cong
    \Q[X]/(\Phi_q(X)),
    \qquad
    \zeta_q=X\bmod\Phi_q(X),
\]
together with the integral basis
\[
    1,\zeta_q,\ldots,\zeta_q^{\varphi(q)-1}
\]
of \(\OO_L\).

For the maximal real subfield, let
\[
    \theta=\zeta_q+\zeta_q^{-1},
\]
and let \(\Psi_q(Y)\in\Z[Y]\) be the monic polynomial characterized by
\[
    \Phi_q(X)
    =
    X^{\varphi(q)/2}\Psi_q(X+X^{-1}).
\]
Then
\[
    L^+
    \cong
    \Q[Y]/(\Psi_q(Y)),
    \qquad
    \theta=Y\bmod\Psi_q(Y),
\]
and
\[
    \OO_{L^+}=\Z[\theta].
\]
Thus
\[
    1,\theta,\ldots,\theta^{\varphi(q)/2-1}
\]
is an integral basis of \(\OO_{L^+}\); see
\cite[Chapter~2]{Washington1997}. The two defining polynomials and the two
integral bases can be constructed from \(q\) in \(q^{O(1)}\) bit operations
and have total description length \(q^{O(1)}\).

The prime ideal
\[
    \mathfrak p_t=(p,\zeta_q-\eta^t)
\]
may be supplied by this two-generator description. Relative to the integral
basis of \(\OO_L\), one may equivalently use the \(\Z\)-basis
\[
    p,\quad
    \zeta_q-\eta^t,\quad
    \zeta_q^2-\eta^{2t},\quad
    \ldots,\quad
    \zeta_q^{\varphi(q)-1}
        -\eta^{(\varphi(q)-1)t}.
\]

\subsection{Reducing Compact Representations}

The factors in a compact representation may contribute positive or
negative powers of \(\mathfrak p\), even when those powers cancel in the
represented algebraic integer. The next lemma makes this cancellation
effective without expanding the product.

\begin{lemma}[Reduction of a compact representation]
\label{lem:compact-reduction}
Let \(h\in\{1,2\}\), and let
\[
    \alpha=\prod_{i=1}^{m}\beta_i^{z_i}\in\OO_L
\]
be nonzero and given by an exact compact representation. Suppose that
\(\mathfrak p\nmid(\alpha)\). Then
\(\rho_{\mathfrak p^h}(\alpha)\) can be computed deterministically in
time polynomial in \(q\), \(\log p\), and the size of the compact
representation.
\end{lemma}

\begin{proof}
Set \(d=\varphi(q)\). For each \(i\), let \(D_i>0\) be the least common
denominator of the coordinates of \(\beta_i\), and write
\[
    \beta_i=\frac{A_i}{D_i},
    \qquad
    A_i=\sum_{j=0}^{d-1}c_{i,j}\zeta_q^j\in\OO_L.
\]
Set
\[
    M_i=\sum_{j=0}^{d-1}|c_{i,j}|,
    \qquad
    C=1+\max_i d\left\lceil\log_2 M_i\right\rceil.
\]
Let \(r_i\) be the exponent of \(\mathfrak p\) in \((A_i)\). Since
\(\Norm(\mathfrak p)=p\),
\[
    p^{r_i}
    \le
    \left|\Norm_{L/\Q}(A_i)\right|
    \le
    M_i^d,
\]
where the second inequality follows because every conjugate of \(A_i\) has
absolute value at most \(M_i\). Hence \(r_i<C\).

Set \(N=C+h\), lift \(\eta\) to \(\eta_N\) modulo \(p^N\), and compute
\[
    y_i
    =
    \sum_{j=0}^{d-1}c_{i,j}\eta_N^j
    \pmod{p^N}.
\]
For every \(0\le r\le N\), the reduction of \(\eta_N\) modulo \(p^r\) is
\(\eta_r\). Since \(\ker(\rho_{\mathfrak p^r})=\mathfrak p^r\),
\[
    y_i\equiv0\pmod{p^r}
    \quad\Longleftrightarrow\quad
    A_i\in\mathfrak p^r.
\]
Thus \(r_i\) is exactly the exponent of \(p\) dividing \(y_i\).

Write
\[
    D_i=p^{s_i}D_i',
    \qquad
    p\nmid D_i'.
\]
Since \(\mathfrak p\) occurs with exponent one in \((p)\), its exponent
in the fractional ideal \((\beta_i)\) is
\[
    e_i=r_i-s_i.
\]
Moreover, \(N-r_i>h\), so division by \(p^{r_i}\) determines
\(y_i/p^{r_i}\) modulo \(p^h\). Define
\[
    u_i
    =
    \frac{y_i}{p^{r_i}}(D_i')^{-1}
    \pmod{p^h}.
\]
The element \(p^{-e_i}\beta_i\) has \(\mathfrak p\)-exponent zero.
By the construction of \(u_i\) and the preceding extension,
\[
    \rho_{\mathfrak p^h}\!\left(p^{-e_i}\beta_i\right)
    =
    u_i.
\]Since exponents
in principal ideals add under multiplication, the exponent of
\(\mathfrak p\) in \((\alpha)\) is
\[
    \sum_{i=1}^{m}z_i e_i=0.
\]
Consequently,
\[
    \alpha
    =
    \prod_{i=1}^{m}\left(p^{-e_i}\beta_i\right)^{z_i},
\]
and therefore
\[
    \rho_{\mathfrak p^h}(\alpha)
    =
    \prod_{i=1}^{m}u_i^{z_i}
    \pmod{p^h}.
\]
Each \(u_i\) is invertible modulo \(p^h\), so negative exponents are handled
by modular inversion.

Finally, \(C\) is polynomially bounded in \(q\) and in the size of the exact
factor descriptions, and \(p^N\) has polynomial bit length. Lifting
\(\eta\) modulo \(p^N\), carrying out the evaluations above, and performing
the final modular exponentiations therefore have
the claimed polynomial complexity.
\end{proof}

\subsection{Fast Interval Products}

We use the following standard baby-step/giant-step product routine.

\begin{lemma}[Bostan--Gaudry--Schost \cite{BGS}]
\label{lem:interval-products}
Let \(R\) be either \(\F_p\) or \(\Z/p^2\Z\). Given \(A\in R\) and an
integer \(0\le M<p\), the interval product
\[
    \prod_{j=1}^{M}(A+j)
\]
can be computed using
\[
    \softO(\sqrt M)
\]
operations in \(R\).
\end{lemma}

The starting point \(A\) is arbitrary, so the lemma applies to every interval
of \(M\) consecutive elements. The construction uses a product tree and fast
multipoint evaluation. Since the polynomial divisions in the remainder tree
are by monic polynomials, the same procedure applies over
\(\Z/p^2\Z\) as over \(\F_p\).

\subsection{Algorithmic Black Boxes}

The Jacobi-sum reconstruction uses three external quantum algorithms. The
arbitrary-degree unit-group algorithm is due to Eisentr\"ager, Hallgren,
Kitaev, and Song \cite{EisentragerHallgrenKitaevSong2014}. Biasse and Song
extended this framework to \(S\)-unit groups and the principal ideal problem
in their SODA paper \cite{BiasseSong2016}. We use the detailed version
\cite{BiasseSong2025} for the exact input-output conventions and quantitative
bounds stated below.

\subsubsection{Computing \(S\)-Unit Groups}

\begin{theorem}[Biasse--Song]\label{thm:s-unit-black-box}
The input consists of
\[
    (F,\OO_F,S),
\]
where \(F\) is supplied by a defining polynomial, \(\OO_F\) is supplied by
an integral basis, and \(S\) is a finite list of prime ideals of \(\OO_F\).
A bounded-error quantum algorithm computes compact generators of the
\(S\)-unit group
\[
    \OO_{F,S}^\times.
\]
Its running time is polynomial in
\[
    [F:\Q],
    \qquad
    \log|\Delta_F|,
    \qquad
    |S|,
    \qquad
    \max_{\mathfrak q\in S}\log\Norm(\mathfrak q).
\]
\end{theorem}

This is \cite[Theorem~1]{BiasseSong2025}. When \(S=\varnothing\), the output
describes the ordinary unit group \(\OO_F^\times\) by its roots of unity and
a compact fundamental-unit basis. In the main algorithm, we apply the theorem
to
\[
    (L,\OO_L,\varnothing)
    \qquad\text{and}\qquad
    (L^+,\OO_{L^+},\varnothing).
\]

\subsubsection{The Principal Ideal Problem}

\begin{theorem}[Biasse--Song]\label{thm:pip-black-box}
The input consists of an explicitly presented ring of integers \(\OO_F\)
and an ideal
\[
    \mathfrak a\subseteq\OO_F.
\]
If \(\mathfrak a\) is principal, a bounded-error quantum algorithm returns a
compact representation of an element \(\gamma\in F\) satisfying
\[
    (\gamma)=\mathfrak a.
\]
Its running time is polynomial in the degree and logarithmic discriminant of
\(F\), and in the bit size of the input ideal.
\end{theorem}

This is \cite[Corollary~1 and Algorithm~6]{BiasseSong2025}. In our application,
the field is \(L\), the ring of integers is \(\OO_L\), and the input ideal is
supplied through the Stickelberger factorization
\[
    \bigl(J_a(\chi)\bigr)
    =
    \prod_t\mathfrak p_t^{\left\lfloor at/q\right\rfloor}.
\]
If the implementation requires an HNF basis, this factored description is
converted by repeated ideal multiplication and polynomial-time HNF
reduction. The norm bound in Step~1 below shows that the resulting basis has
polynomial bit length.

\subsubsection{Gauss-Sum Phase Estimation}

\begin{theorem}[Van Dam--Seroussi]
\label{thm:gauss-phase-black-box}
Let \(\psi\) be a nontrivial multiplicative character of
\(\F_p^\times\), supplied by a succinct description that allows
\(\psi(x)\) to be evaluated coherently and efficiently. Given
\[
    0<\varepsilon<1
    \qquad\text{and}\qquad
    0<\delta<\frac12,
\]
a bounded-error quantum algorithm returns an angle
\(\widetilde\theta\) such that, with probability at least \(1-\delta\),
\[
    \min_{k\in\Z}
    \left|
        \widetilde\theta-\arg G(\psi)+2\pi k
    \right|
    \le \varepsilon.
\]
Its running time is
\[
    O\!\left(
        \varepsilon^{-1}
        \polylog(p)
        \log(1/\delta)
    \right).
\]
\end{theorem}

The constant-success-probability phase-estimation algorithm is due to
Van Dam and Seroussi \cite{VanDamSeroussi}; the dependence on \(\delta\)
follows by repetition and amplification. In the main algorithm, we apply it
to
\[
    \chi
    \qquad\text{and}\qquad
    \chi^a.
\]

\section{The Prime-Modulus Algorithm}\label{sec:prime-modulus}

The goal is not merely to approximate \(J_a(\chi)\) as a complex number,
but to recover it in a compact algebraic form that can be reduced modulo the
prime ideal \(\mathfrak p\). We use the notation and algorithmic interfaces
introduced in \Cref{sec:preliminaries}.

\subsection{Algorithm Overview}

\paragraph{Input.}
A prime \(p\), a divisor \(q\mid p-1\), an integer \(1\le a<q\), and the
compatible pair \((\chi,\mathfrak p)\) defined in
\Cref{sec:preliminaries}.

\paragraph{Output.}
An exact compact representation of \(J_a(\chi)\), together with its
reductions modulo \(\mathfrak p\) and \(\mathfrak p^2\).

\begin{enumerate}

\item \textsc{[Classical]} \textbf{Compute the Stickelberger factorization.}

Using \Cref{prop:stickelberger-factorization}, compute
\[
    \bigl(J_a(\chi)\bigr)
    =
    \prod_{t\in(\Z/q\Z)^\times}
    \mathfrak p_t^{e_t},
    \qquad
    e_t=\left\lfloor\frac{at}{q}\right\rfloor.
\]

\item \textsc{[Quantum]} \textbf{Compute a generator of the principal ideal.}

Apply \Cref{thm:pip-black-box} to
\[
    \left(
        \OO_L,\,
        \prod_t\mathfrak p_t^{e_t}
    \right),
\]
with the factorization
\[
    \{(\mathfrak p_t,e_t)\}_t
\]
supplied explicitly. Obtain a compactly represented element \(\gamma\)
satisfying
\[
    (\gamma)=\bigl(J_a(\chi)\bigr).
\]
Thus,
\[
    \gamma=uJ_a(\chi)
\]
for some unknown unit \(u\in\OO_L^\times\).

\item \textsc{[Quantum]} \textbf{Compute fundamental-unit bases.}

Apply \Cref{thm:s-unit-black-box} with \(S=\varnothing\) to
\[
    (L,\OO_L,\varnothing)
    \qquad\text{and}\qquad
    (L^+,\OO_{L^+},\varnothing).
\]
Obtain compact fundamental-unit bases
\[
    \OO_L^\times
    =
    \mu(L)\times
    \langle\epsilon_1,\ldots,\epsilon_r\rangle
\]
and
\[
    \OO_{L^+}^\times
    =
    \{\pm1\}\times
    \langle e_1,\ldots,e_r\rangle.
\]

\item \textsc{[Classical]} \textbf{Solve the relative norm equation.}

Set
\[
    v=\frac{p^{a-1}}{\gamma\overline\gamma}
    \in\OO_{L^+}^\times.
\]
Using logarithmic embeddings, recover the free-unit coordinates
\[
    v=\nu_0\prod_{i=1}^r e_i^{b_i},
    \qquad
    \epsilon_j\overline{\epsilon_j}
    =
    \nu_j\prod_{i=1}^r e_i^{A_{ij}},
\]
where \(\nu_0,\nu_j\in\{\pm1\}\). Let
\[
    b=(b_1,\ldots,b_r)^T,
\]
and let \(A\in\Z^{r\times r}\) be the matrix whose \(j\)-th column is the
exponent vector of \(\epsilon_j\overline{\epsilon_j}\) in the basis
\(e_1,\ldots,e_r\).

Solve
\[
    Ax=b
\]
over the integers. For a solution \(x=(x_1,\ldots,x_r)^T\), set
\[
    w=\prod_{j=1}^r\epsilon_j^{x_j},
    \qquad
    \alpha=w\gamma.
\]
Then
\[
    w\overline w=v,
    \qquad
    \alpha=\xi J_a(\chi)
\]
for some \(\xi\in\mu(L)\).

\item \textsc{[Quantum]} \textbf{Identify the missing phase.}

By \Cref{prop:gauss-jacobi},
\[
    J_a(\chi)=\frac{G(\chi)^a}{G(\chi^a)}.
\]
Apply \Cref{thm:gauss-phase-black-box} to \(\chi\) and \(\chi^a\), and
compute phase estimates
\[
    \theta_1\approx\arg G(\chi),
    \qquad
    \theta_a\approx\arg G(\chi^a).
\]
Set
\[
    \theta_J
    =
    a\theta_1-\theta_a
    \pmod{2\pi},
\]
compute
\[
    \theta_\alpha\approx\arg\alpha,
\]
and then set
\[
    \theta_\xi
    =
    \theta_\alpha-\theta_J
    \pmod{2\pi}.
\]
Enumerate the roots of unity in \(\mu(L)\), and select the unique
\(\xi\in\mu(L)\) whose argument is closest to \(\theta_\xi\). This gives
\[
    J_a(\chi)=\xi^{-1}\alpha.
\]

\end{enumerate}

The identity
\[
    J_a(\chi)=\xi^{-1}\alpha
\]
is the desired exact compact representation. The Stickelberger factorization
shows that the exponent of \(\mathfrak p\) in
\(\bigl(J_a(\chi)\bigr)\) is
\[
    \left\lfloor\frac aq\right\rfloor=0.
\]
For \(h\in\{1,2\}\), apply \Cref{lem:compact-reduction} to obtain
\[
    \rho_{\mathfrak p^h}\!\left(J_a(\chi)\right)
    =
    \rho_{\mathfrak p^h}\!\left(\xi^{-1}\alpha\right).
\]

\section{Correctness}\label{sec:correctness}

\begin{theorem}\label{thm:jacobi-algorithm-correctness}
With bounded error, the algorithm above outputs an exact compact
representation of \(J_a(\chi)\) and its reductions modulo
\(\mathfrak p\) and \(\mathfrak p^2\).
\end{theorem}

\begin{proof}
We verify the five steps in order.

\paragraph{Step 1: Stickelberger factorization.}

By \Cref{prop:stickelberger-factorization},
\[
    \bigl(J_a(\chi)\bigr)
    =
    \prod_{t\in(\Z/q\Z)^\times}
    \mathfrak p_t^{\left\lfloor at/q\right\rfloor}.
\]
Thus Step~1 constructs the prime-ideal factorization of the principal ideal
generated by \(J_a(\chi)\).

\paragraph{Step 2: Principal-ideal recovery.}

The input ideal is principal by Step~1. Therefore,
\Cref{thm:pip-black-box} returns \(\gamma\) satisfying
\[
    (\gamma)=\bigl(J_a(\chi)\bigr).
\]
The quotient
\[
    u=\frac{\gamma}{J_a(\chi)}
\]
is taken in \(L\). Since \(\gamma\) and \(J_a(\chi)\) generate the same
principal ideal, both \(u\) and \(u^{-1}\) belong to \(\OO_L\). Hence
\[
    u\in\OO_L^\times
\]
and
\[
    \gamma=uJ_a(\chi).
\]

\paragraph{Step 3: Fundamental-unit bases.}

By \Cref{thm:s-unit-black-box}, the two calls with \(S=\varnothing\)
return fundamental-unit bases for
\[
    \OO_L^\times
    \qquad\text{and}\qquad
    \OO_{L^+}^\times.
\]
These unit groups have the same free rank, so the relative norm map is
represented in the chosen bases by the square matrix \(A\).

\paragraph{Step 4: The relative norm equation.}

By \Cref{prop:gauss-jacobi} and the identity
\(\gamma=uJ_a(\chi)\),
\[
    v
    =
    \frac{p^{a-1}}{\gamma\overline\gamma}
    =
    \frac{1}{u\overline u}.
\]
Thus \(v\) is a totally positive unit of \(\OO_{L^+}\), and \(u^{-1}\)
is a solution of
\[
    z\overline z=v.
\]

Write
\[
    u^{-1}
    =
    \xi_0
    \prod_{j=1}^r
    \epsilon_j^{x_j^{(0)}}
\]
for some \(\xi_0\in\mu(L)\). Taking relative norms and expressing the result
in the basis \(e_1,\ldots,e_r\) gives
\[
    Ax^{(0)}=b.
\]
Therefore, the integer system \(Ax=b\) is consistent.

Let \(x\) be any integer solution and set
\[
    w=\prod_{j=1}^r\epsilon_j^{x_j}.
\]
The equality \(Ax=b\) implies that \(w\overline w\) and \(v\) have the same
free-unit coordinates. Hence
\[
    \frac{w\overline w}{v}\in\{\pm1\}.
\]
Both \(w\overline w\) and \(v\) are totally positive, so this quotient is
\(1\). Therefore,
\[
    w\overline w=v.
\]

Now set
\[
    \alpha=w\gamma.
\]
Since \(w\) is a unit,
\[
    (\alpha)=\bigl(J_a(\chi)\bigr),
\]
and
\[
    \alpha\overline\alpha
    =
    w\overline w\,
    \gamma\overline\gamma
    =
    p^{a-1}.
\]
Let
\[
    \delta=\frac{\alpha}{J_a(\chi)}.
\]
The equality of principal ideals gives
\[
    \delta\in\OO_L^\times,
\]
while
\[
    \delta\overline\delta
    =
    \frac{\alpha\overline\alpha}
         {J_a(\chi)\overline{J_a(\chi)}}
    =
    1.
\]
For every complex embedding
\(\sigma\colon L\hookrightarrow\mathbb C\),
\[
    |\sigma(\delta)|^2
    =
    \sigma(\delta\overline\delta)
    =
    1.
\]
By Kronecker's theorem, \(\delta\) is a root of unity. Thus
\[
    \alpha=\xi J_a(\chi)
\]
for some \(\xi\in\mu(L)\).

\paragraph{Step 5: Identification of the missing phase.}

For \(a=1\), one has \(J_1(\chi)=1\), so suppose \(2\le a<q\). By
\Cref{prop:gauss-jacobi},
\[
    \arg J_a(\chi)
    \equiv
    a\,\arg G(\chi)-\arg G(\chi^a)
    \pmod{2\pi}.
\]
Choose the Gauss-sum estimates so that the resulting approximation to
\(\arg J_a(\chi)\) has circular error less than
\[
    \frac{\pi}{8q},
\]
and evaluate \(\arg\alpha\) to circular error less than the same quantity.
Since
\[
    \alpha=\xi J_a(\chi),
\]
the difference of the two estimates approximates \(\arg\xi\) to error less
than
\[
    \frac{\pi}{4q}.
\]

Every root of unity in \(L\) belongs to \(\mu_{2q}\). Distinct candidates
are therefore separated by an angle of at least
\[
    \frac{\pi}{q}.
\]
The unique closest root of unity is consequently the correct \(\xi\), and
\[
    J_a(\chi)=\xi^{-1}\alpha.
\]

\paragraph{Reduction modulo \(\mathfrak p^h\).}

The exponent of the distinguished prime
\(\mathfrak p=\mathfrak p_1\) in the Stickelberger factorization is
\[
    \left\lfloor\frac{a}{q}\right\rfloor=0.
\]
Hence \(\mathfrak p\) does not divide \(\bigl(J_a(\chi)\bigr)\). For each
\(h\in\{1,2\}\), the powers of \(\mathfrak p\) contributed by the
individual factors in the compact representation therefore cancel, and
\Cref{lem:compact-reduction} gives
\[
    \rho_{\mathfrak p^h}\!\left(\xi^{-1}\alpha\right)
    =
    \rho_{\mathfrak p^h}\!\left(J_a(\chi)\right).
\]

The quantum subroutines in Steps~2, 3, and~5 have bounded error. Amplifying
each to a sufficiently small constant failure probability and applying a
union bound gives an overall success probability bounded away from
\(1/2\).
\end{proof}

\subsection{Time and Space Complexity}

Set
\[
    d=\varphi(q),
    \qquad
    \mathcal{B}=\left\lceil\log_2p\right\rceil,
\]
so that \(\mathcal{B}\) is the bit length of \(p\). We count classical bit
operations, quantum gates, classical space, and quantum space.
Throughout this subsection, \(\softO\) suppresses factors polynomial in
\(\log p\) and \(\log q\), together with the cost of constant-error
amplification.

\paragraph{Input descriptions.}

The field \(L\) has degree
\[
    [L:\Q]=d\le q,
\]
and the field \(L^+\) has degree \(d/2\). Their computational
representations from \Cref{sec:preliminaries} have bit length
\(q^{O(1)}\). Moreover,
\[
    \log|\Delta_L|=O(q\log q),
    \qquad
    \log|\Delta_{L^+}|=O(q\log q).
\]

The character \(\chi\) is represented by \(p,q\), a primitive root
\(g\in\F_p^\times\), and the rule
\[
    \chi(g^m)=\zeta_q^m.
\]
Shor's algorithm factors \(p-1\) and computes discrete logarithms in
\(\F_p^\times\) in quantum time polynomial in \(\mathcal{B}\)
\cite{Shor}. Once the prime divisors of \(p-1\) are known, a primitive root
can be found by sampling and testing candidates. The discrete-logarithm
circuit and the subsequent arithmetic can be implemented reversibly, giving
the coherent character-evaluation oracles required by
\Cref{thm:gauss-phase-black-box}. Constructing \(g\), evaluating \(\chi\)
and \(\chi^a\), and computing
\[
    \eta=g^{-K},
    \qquad
    \mathfrak p=(p,\zeta_q-\eta)
\]
therefore require only \(\poly(\mathcal{B})\) quantum gates and classical bit
operations.

\paragraph{Step 1: Stickelberger factorization.}

There are
\[
    d=\varphi(q)\le q
\]
values of \(t\in(\Z/q\Z)^\times\). For each \(t\), the algorithm computes
\[
    \eta^t\bmod p,
    \qquad
    e_t=\left\lfloor\frac{at}{q}\right\rfloor.
\]
The complete list is constructed in
\[
    \softO(q\mathcal{B})
\]
classical bit operations.

Each pair \((\mathfrak p_t,e_t)\) requires
\[
    O(\mathcal{B}+\log q)
\]
bits. Thus the factored ideal
\[
    \bigl(J_a(\chi)\bigr)
    =
    \prod_t\mathfrak p_t^{e_t}
\]
is represented using
\[
    O\bigl(q(\mathcal{B}+\log q)\bigr)
\]
bits.

Every \(\mathfrak p_t\) has norm \(p\). Pairing \(t\) with \(q-t\) gives
\[
    \sum_t e_t
    =
    \frac{(a-1)\varphi(q)}{2}.
\]
Consequently,
\[
    \log\Norm\!\left(\bigl(J_a(\chi)\bigr)\right)
    =
    \frac{(a-1)\varphi(q)}{2}\log p
    =
    O(q^2\mathcal{B}).
\]

\paragraph{Step 2: Principal-ideal recovery.}

Apply the complexity guarantee in
\Cref{thm:pip-black-box}. In our setting,
\[
    [L:\Q]\le q,
    \qquad
    \log|\Delta_L|=O(q\log q),
\]
and the input ideal has bit length polynomial in \(q\) and \(\mathcal{B}\).
Therefore, \(\gamma\) is computed in
\[
    q^{O(1)}\mathcal{B}^{O(1)}
\]
quantum time, using polynomial classical and quantum workspace. Its compact
representation has polynomial bit length.

\paragraph{Step 3: Fundamental-unit bases.}

Apply the complexity guarantee in
\Cref{thm:s-unit-black-box} with \(S=\varnothing\). The two fields have
degrees \(d\) and \(d/2\), and both logarithmic discriminants are
\(O(q\log q)\). Hence the two fundamental-unit bases are computed in
\[
    q^{O(1)}
\]
quantum time and polynomial classical and quantum workspace.

Each basis contains
\[
    r=\frac{\varphi(q)}{2}-1=O(q)
\]
fundamental units, each returned in compact representation.

\paragraph{Step 4: The relative norm equation.}

The compact representations of \(\gamma\), the \(\epsilon_j\), and the
\(e_i\) can be evaluated in every archimedean embedding to polynomial
precision in polynomial time. Let
\[
    \sigma_1,\ldots,\sigma_{r+1}\colon
    L^+\longrightarrow\R
\]
be the real embeddings of \(L^+\). For a unit
\(y\in\OO_{L^+}^\times\), define its logarithmic embedding by
\[
    \ell(y)
    :=
    \left(
        \log|\sigma_1(y)|,\ldots,
        \log|\sigma_{r+1}(y)|
    \right)^T.
\]
Since
\[
    \left|\Norm_{L^+/\Q}(y)\right|=1,
\]
the coordinates of \(\ell(y)\) sum to zero. We may therefore delete one coordinate and regard
\[
    \ell(y)\in\R^r.
\]

When \(r=0\), there are no free-unit coordinates to recover, and the
coordinate-recovery computation is vacuous. For the remainder of this
argument, assume \(r\ge1\).

Let \(E\in\R^{r\times r}\) be the matrix whose \(i\)-th column is
\(\ell(e_i)\). If
\[
    y=\nu\prod_{i=1}^r e_i^{c_i},
    \qquad
    \nu\in\{\pm1\},
\]
and
\[
    c=(c_1,\ldots,c_r)^T,
\]
then
\[
\begin{aligned}
    \ell(y)
    &=
    \sum_{i=1}^r c_i\ell(e_i) = Ec.
\end{aligned}
\]
Thus the exponent vector of \(y\) is
\[
    c=E^{-1}\ell(y).
\]

The columns of \(E\) generate the logarithmic unit lattice after one
coordinate has been deleted. Before this deletion, the first minimum of the
lattice is at least
\[
    \frac{\log(d/2)}{6(d/2)^4}
\]
by \cite[Proposition~6]{BiasseSong2025}. Since the deleted coordinate is the
negative of the sum of the remaining coordinates, deleting it decreases
Euclidean lengths by at most a factor of \(\sqrt{r+1}\). Hence every
nonzero vector in \(E\Z^r\) has length at least
\[
    \frac{\log(d/2)}
         {6\sqrt{r+1}(d/2)^4}
    =
    q^{-O(1)}.
\]
The balls of half this radius centered at the lattice points are disjoint,
so a standard packing argument gives
\[
    |\det E|
    \ge
    2^{-q^{O(1)}}.
\]

The exact compact representations also bound the size of the entries of
\(E\) and \(\ell(y)\). Indeed, after clearing denominators in each compact
factor, its coefficient size bounds all of its conjugates from above, while
its nonzero norm bounds them from below. Combining these bounds with the
exponents in the compact representation gives
\[
    \max\bigl\{
        \|E\|_2,\,
        \|\ell(y)\|_2
    \bigr\}
    \le
    2^{q^{O(1)}\mathcal B^{O(1)}}
\]
for every unit \(y\) considered here. Cramer's rule, Hadamard's inequality,
and the lower bound on \(|\det E|\) therefore give
\[
    \|E^{-1}\|_2
    \le
    2^{q^{O(1)}\mathcal B^{O(1)}}.
\]
It follows from
\[
    c=E^{-1}\ell(y)
\]
that the coordinates of \(c\) have polynomial bit length.

It remains to justify their exact recovery from numerical approximations.
The compact representations allow the entries of \(E\) and \(\ell(y)\) to
be computed with absolute error at most \(2^{-T}\) in time polynomial in
\(q\), \(\mathcal B\), and \(T\). Let \(\widetilde E\) and
\(\widetilde\ell\) be the resulting approximations, and set
\[
    \widetilde c
    =
    \widetilde E^{-1}\widetilde\ell.
\]
For polynomially large \(T\), the preceding bounds ensure that
\[
    \|E^{-1}\|_2
    \|\widetilde E-E\|_2
    \le
    \frac12,
\]
and therefore
\[
    \|\widetilde E^{-1}\|_2
    \le
    2\|E^{-1}\|_2.
\]
Moreover,
\[
    \widetilde c-c
    =
    \widetilde E^{-1}
    \left(
        \widetilde\ell-\ell(y)
        -
        (\widetilde E-E)c
    \right).
\]
Polynomially many bits of precision consequently make
\[
    \|\widetilde c-c\|_2<\frac12.
\]
Rounding each coordinate of \(\widetilde c\) therefore recovers \(c\)
exactly.

Applying this procedure to
\[
    y=v
    \qquad\text{and}\qquad
    y=\epsilon_j\overline{\epsilon_j},
    \quad 1\le j\le r,
\]
computes the vector \(b\) and every column of \(A\). Their entries have
polynomial bit length, and the complete coordinate-recovery computation has
\[
    q^{O(1)}\mathcal B^{O(1)}
\]
bit complexity. The system
\[
    Ax=b
\]
can therefore be solved by Smith normal form in
\[
    q^{O(1)}\mathcal{B}^{O(1)}
\]
bit operations and workspace \cite{KannanBachem1979}.

The resulting exponent vector has polynomial bit length, so
\[
    w=\prod_{j=1}^r\epsilon_j^{x_j}
\]
is retained as a compact power product. Forming
\[
    \alpha=w\gamma
\]
has the same polynomial bit complexity.

\paragraph{Step 5: Identification of the missing phase.}

Distinct elements of \(\mu(L)\subseteq\mu_{2q}\) are separated by an angle
of at least \(\pi/q\). It is enough to estimate
\[
    \arg G(\chi)
\]
to error \(O((aq)^{-1})\), estimate
\[
    \arg G(\chi^a)
\]
to error \(O(q^{-1})\), and evaluate
\[
    \arg\alpha
\]
to error \(O(q^{-1})\).

By \Cref{thm:gauss-phase-black-box}, the two Gauss-sum estimates require
\[
    \softO(aq+q)
\]
quantum gates. Since \(a<q\), this is at most
\[
    \softO(q^2).
\]
Evaluating the phase of \(\alpha\) and enumerating the at most \(2q\) roots
of unity require
\[
    q^{O(1)}\mathcal{B}^{O(1)}
\]
classical time and workspace.

\paragraph{Reduction modulo \(\mathfrak p^h\).}

The final compact representation of \(J_a(\chi)\) has size polynomial in
\(q\) and \(\mathcal B\). In the Stickelberger factorization, the exponent of
the distinguished prime \(\mathfrak p=\mathfrak p_1\) is
\[
    \left\lfloor\frac aq\right\rfloor=0.
\]
Thus \(\mathfrak p\) does not divide \(\bigl(J_a(\chi)\bigr)\), and for each
\(h\in\{1,2\}\), \Cref{lem:compact-reduction} computes
\[
    \rho_{\mathfrak p^h}\!\left(J_a(\chi)\right)
    \in
    \Z/p^h\Z
\]
in
\[
    q^{O(1)}\mathcal B^{O(1)}
\]
classical bit operations and workspace.

\subsection{Proof of the Jacobi-Sum Computation Theorem}

\begin{proof}[Proof of \Cref{thm:jacobi-computation}]
By \Cref{thm:jacobi-algorithm-correctness}, the algorithm outputs an exact
compact representation of \(J_a(\chi)\), together with its reductions modulo
\(\mathfrak p\) and \(\mathfrak p^2\), with bounded error.

Step~1, Step~4, and the reductions modulo \(\mathfrak p^h\) have
\[
    q^{O(1)}\mathcal{B}^{O(1)}
\]
classical bit complexity. Steps~2 and~3 have
\[
    q^{O(1)}\mathcal{B}^{O(1)}
\]
quantum gate complexity by
\Cref{thm:pip-black-box} and
\Cref{thm:s-unit-black-box}. Step~5 costs
\[
    \softO(aq+q),
\]
which is \(\softO(q^2)\) uniformly for \(a<q\).

Choosing \(c\) larger than the fixed degrees of these polynomial bounds
gives the expected running time
\[
    \softO(q^c).
\]
The exact compact representation of \(J_a(\chi)\), as well as all
intermediate compact representations and matrices, has
\(q^{O(1)}\polylog(p)\) bit length. The same form of bound therefore applies
to classical and quantum workspace. Constant-error amplification changes
the complexity only by factors hidden in \(\softO\).
\end{proof}

\subsection{Proof of the Main Theorem}

We now combine the Jacobi-sum procedure with two classical interval
products.

\begin{proof}[Proof of \Cref{thm:main}]
First suppose that \(0<n<p-1\), and write
\[
    n=aK+r,
    \qquad
    0\le a<q,
    \qquad
    0\le r<K.
\]

If \(a=0\), then \(n=r<K\), and
\Cref{lem:interval-products} computes \(n!\bmod p\) directly in
\[
    \softO(\sqrt K)
    =
    \softO\left(\sqrt{\frac{p}{q}}\right)
\]
time.

Now suppose that \(1\le a<q\). Combining \eqref{eq:central} with the identity
\(n!=(aK)!\prod_{j=1}^{r}(aK+j)\) gives
\[
    n!
    \equiv
    (-1)^{a-1}
    \rho_{\mathfrak p}\!\left(J_a(\chi)\right)
    (K!)^a
    \prod_{j=1}^{r}(aK+j)
    \pmod p.
\]

The two interval products
\[
    K!
    \qquad\text{and}\qquad
    \prod_{j=1}^{r}(aK+j)
\]
have lengths at most \(K\). By
\Cref{lem:interval-products}, they are computed modulo \(p\) in
\[
    \softO(\sqrt K)
    =
    \softO\left(\sqrt{\frac{p}{q}}\right)
\]
time. Raising \(K!\) to the power \(a\) requires only
\(\polylog(p)\) additional bit operations.

By \Cref{thm:jacobi-computation},
\[
    \rho_{\mathfrak p}\!\left(J_a(\chi)\right)
\]
is computed with bounded error in expected time
\[
    \softO(q^c).
\]

The endpoint \(n=p-1\) is handled by writing
\[
    p-1=(q-1)K+K.
\]
The central congruence gives
\[
    ((q-1)K)!\pmod p,
\]
after which one final interval product of length \(K\) remains.

Combining the Jacobi-sum computation with the interval products gives total
expected running time
\[
    \softO\left(
        q^c+\sqrt{\frac{p}{q}}
    \right).
\]
The classical steps are deterministic, and the failure probability of the
quantum steps is bounded by the preceding analysis. Hence the resulting
algorithm is a bounded-error quantum algorithm computing
\[
    n!\pmod p.
\]
\end{proof}

\section{Prime-Square Moduli}\label{sec:prime-square}

The prime-modulus algorithm reconstructs \(J_a(\chi)\) in compact algebraic
form before reducing it modulo \(\mathfrak p\). The same compact
representation can instead be reduced modulo \(\mathfrak p^2\). Using
notation from \Cref{sec:preliminaries}, define
\[
    J_a^{(2)}
    :=
    \rho_{\mathfrak p^2}\!\left(J_a(\chi)\right)
    \in\Z/p^2\Z.
\]
For \(m\ge0\), let
\[
    H_m
    :=
    \sum_{j=1}^{m}\frac1j
    \pmod p,
    \qquad
    H_0:=0.
\]
All inverses in \(H_m\) are taken in \(\F_p\). Since \(q\mid p-1\), the
integer \(q\) is invertible modulo \(p^2\).

\subsection{The Central Congruence Modulo \texorpdfstring{\(p^2\)}{p squared}}

\begin{proposition}\label{prop:lifted-central}
For every \(1\le a<q\),
\[
    (aK)!
    \equiv
    (-1)^{a-1}J_a^{(2)}(K!)^a
    \left(
        1+\frac{ap}{q}(H_{aK}-H_K)
    \right)
    \pmod{p^2}.
\]
\end{proposition}

The proof is given in \Cref{app:lifted-central}. It follows from a
specialization of Young's congruence between Jacobi sums and multinomial
coefficients \cite[Theorem~2.2]{Young1995}, followed by an elementary block
expansion modulo \(p^2\).

\subsection{Evaluating the Harmonic Term}

It remains to compute
\[
    H_{aK}-H_K
\]
without summing an interval whose length may be comparable to \(p\).

For \(1\le t<q\), define
\[
    Q_t
    :=
    \frac{(1-\widetilde\eta^t)^{p-1}-1}{p}
    \pmod p.
\]
The power is computed modulo \(p^2\). Since
\[
    1-\eta^t\neq0
    \qquad\text{in }\F_p,
\]
Fermat's theorem implies that the numerator is divisible by \(p\).

\begin{lemma}\label{lem:fourier-harmonic}
For every \(1\le a<q\),
\[
    H_{aK}-H_K
    \equiv
    \sum_{t=1}^{q-1}
    (\eta^{at}-\eta^t)Q_t
    \pmod p.
\]
\end{lemma}

\begin{proof}
The congruence
\[
    \binom pm
    \equiv
    p\frac{(-1)^{m-1}}{m}
    \pmod{p^2},
    \qquad
    1\le m<p,
\]
and the identity
\[
    \widetilde\eta^p=\widetilde\eta
\]
give, for \(1\le t<q\),
\[
    (1-\widetilde\eta^t)^p
    \equiv
    1-\widetilde\eta^t
    -
    p\sum_{m=1}^{p-1}\frac{\eta^{tm}}{m}
    \pmod{p^2}.
\]
Dividing by \(1-\widetilde\eta^t\) and using the definition of \(Q_t\)
gives
\begin{equation}\label{eq:fermat-fourier}
    \sum_{m=1}^{p-1}\frac{\eta^{tm}}{m}
    \equiv
    -(1-\eta^t)Q_t
    \pmod p.
\end{equation}

Now write
\[
    H_{aK}-H_K
    =
    \sum_{k=1}^{a-1}
    \sum_{u=1}^{K}
    \frac1{kK+u}.
\]
Since
\[
    q(kK+u)
    =
    k(p-1)+qu
    \equiv
    qu-k
    \pmod p,
\]
we have
\[
    \sum_{u=1}^{K}\frac1{kK+u}
    =
    q
    \sum_{\substack{1\le m<p\\m\equiv-k\;(\mathrm{mod}\;q)}}
    \frac1m.
\]
The roots-of-unity identity
\[
    \mathbf 1_{m\equiv-k\;(\mathrm{mod}\;q)}
    =
    \frac1q
    \sum_{t=0}^{q-1}\eta^{t(m+k)}
\]
therefore gives
\[
\begin{aligned}
    H_{aK}-H_K
    &=
    \sum_{k=1}^{a-1}
    \sum_{t=0}^{q-1}
    \eta^{kt}
    \sum_{m=1}^{p-1}\frac{\eta^{tm}}{m}.
\end{aligned}
\]
The term \(t=0\) vanishes because
\[
    H_{p-1}=0\pmod p.
\]
For \(1\le t<q\),
\[
    \sum_{k=1}^{a-1}\eta^{kt}
    =
    \frac{\eta^t-\eta^{at}}{1-\eta^t}.
\]
Using \eqref{eq:fermat-fourier}, we obtain
\[
\begin{aligned}
    H_{aK}-H_K
    &=
    \sum_{t=1}^{q-1}
    \frac{\eta^t-\eta^{at}}{1-\eta^t}
    \left(-(1-\eta^t)Q_t\right) \\
    &=
    \sum_{t=1}^{q-1}
    (\eta^{at}-\eta^t)Q_t
    \pmod p.
\end{aligned}
\]
\end{proof}

\subsection{Proof of the Prime-Square Theorem}

\begin{proof}[Proof of \Cref{thm:prime-square}]
First suppose that \(0\le n<p\).

If \(n<p-1\), write
\[
    n=aK+r,
    \qquad
    0\le a<q,
    \qquad
    0\le r<K.
\]
If \(n=p-1\), set
\[
    a=q-1,
    \qquad
    r=K,
\]
so that
\[
    p-1=aK+r.
\]
Thus the final interval always has length at most \(K\).

If \(a=0\), then \(n<K\), and
\Cref{lem:interval-products} computes \(n!\bmod p^2\) directly. Suppose
now that \(1\le a<q\). The same lemma computes
\[
    K!\pmod{p^2}
    \qquad\text{and}\qquad
    \prod_{j=1}^{r}(aK+j)\pmod{p^2}
\]
in
\[
    \softO\left(\sqrt{\frac pq}\right)
\]
time.

For \(a=1\), one has
\[
    J_1^{(2)}=1.
\]
For \(2\le a<q\), \Cref{thm:jacobi-computation} produces an exact
compact representation of \(J_a(\chi)\) and, taking \(h=2\), computes
\[
    J_a^{(2)}
    =
    \rho_{\mathfrak p^2}\!\left(J_a(\chi)\right)
    \in\Z/p^2\Z
\]
in expected time \(\softO(q^c)\).

By \Cref{lem:fourier-harmonic}, the harmonic term in
\Cref{prop:lifted-central} is obtained from the \(q-1\) values \(Q_t\).
Each value requires one modular exponentiation modulo \(p^2\), so the total
cost is
\[
    \softO(q).
\]
This is absorbed by the \(\softO(q^c)\) cost of reconstructing the Jacobi
sum. Applying \Cref{prop:lifted-central} and multiplying by
\[
    \prod_{j=1}^{r}(aK+j)
\]
gives \(n!\bmod p^2\).

It remains to consider \(n\ge p\). If
\[
    p\le n<2p,
\]
write
\[
    n=p+m,
    \qquad
    0\le m<p.
\]
Then
\[
\begin{aligned}
    (p+m)!
    &=
    p(p-1)!
    \prod_{j=1}^{m}(p+j) \\
    &\equiv
    -p\,m!
    \pmod{p^2}.
\end{aligned}
\]
If \(m=0\), then \(m!=1\). If \(m>0\), only \(m!\bmod p\) is required, and
\Cref{thm:main} computes it within the same asymptotic bound.

If \(n\ge2p\), then both \(p\) and \(2p\) divide \(n!\), so
\[
    n!\equiv0\pmod{p^2}.
\]

Combining all cases gives expected running time
\[
    \softO\left(
        q^c+\sqrt{\frac pq}
    \right),
\]
with bounded error.
\end{proof}

\begin{remark}[Wilson quotients]
Let \(F\) be the representative of \((p-1)!\bmod p^2\) in
\(\{0,\ldots,p^2-1\}\). Wilson's theorem gives
\[
    p\mid F+1,
\]
and hence
\[
    W_p
    \equiv
    \frac{F+1}{p}
    \pmod p.
\]
Thus \Cref{thm:prime-square} computes the Wilson quotient modulo \(p\).
In particular, \(p\) is a Wilson prime if and only if
\[
    F\equiv-1\pmod{p^2}.
\]
\end{remark}

\section{Conclusion and Open Problems}\label{sec:conclusion}

We have given bounded-error quantum algorithms for computing modular
factorials modulo \(p\) and \(p^2\) under a provided-divisor promise on
\(p-1\). The central technical result is the exact reconstruction, in compact
algebraic form, of Jacobi sums attached to characters whose order grows with
the input. This representation can be evaluated modulo both
\(\mathfrak p\) and \(\mathfrak p^2\) without expanding the underlying
algebraic integer.

The first natural extension is from prime-square to higher prime-power
moduli.

\begin{conjecture}
\Cref{thm:prime-square} extends to every modulus \(p^e\), with at most a
\(\poly(e)\) factor in the running time.
\end{conjecture}

The condition \(q\mid(p-1)\) is essential to the present method and is its greatest weakness, but we
suspect that the exact divisibility is not inherent to the problem. As in other
quantum algorithms where a shift overcomes a lack of exact alignment, a
suitably shifted version of our construction may work for an intermediate
\(q\) that does not divide \(p-1\). We therefore conjecture the following.

\begin{conjecture}
There exists an absolute constant \(\varepsilon>0\) and a bounded-error
quantum algorithm that, for every prime \(p\) and every \(0<n<p\), computes
\[
    n!\pmod p
\]
in time
\[
    p^{1/2-\varepsilon+o(1)}.
\]
\end{conjecture}

Even within the studied regime, it would be useful to make the
constant \(c\) explicit and to reduce the polynomial dependence on \(q\).

\section*{Acknowledgments}
\addcontentsline{toc}{section}{Acknowledgments}

The author thanks Professor Amnon Ta-Shma for his valuable advice on the
structure of this paper and for his generous availability and support. The
author also thanks Professor Amir Shpilka for providing a supportive working
environment in his laboratory and for indirectly bringing the problem of
modular factorials to the author's attention.

\clearpage

\printbibliography[heading=bibintoc,title={References}]

\clearpage

\appendix

\section{Proof of the Central Congruence}\label{app:central-congruence}

We prove \eqref{eq:central}. For \(x=g^m\in\F_p^\times\),
\[
    \rho_{\mathfrak p}(\chi(x))
    =
    \rho_{\mathfrak p}(\zeta_q^m)
    =
    \eta^m
    =
    (g^m)^{-K}
    =
    x^{-K}.
\]
Set
\[
    s=p-1-K.
\]
Then \(x^{-K}=x^s\) for \(x\neq0\), while both sides are interpreted as
zero at \(x=0\). It follows that
\[
    \rho_{\mathfrak p}\!\left(J_a(\chi)\right)
    =
    \sum_{\substack{x_1,\ldots,x_a\in\F_p\\
                    x_1+\cdots+x_a=1}}
    \prod_{i=1}^{a}x_i^s.
\]

For \(y\in\F_p\), the indicator of the condition \(y=0\) is
\[
    \mathbf 1_{y=0}=1-y^{p-1}.
\]
Therefore,
\[
    \rho_{\mathfrak p}\!\left(J_a(\chi)\right)
    =
    \sum_{x_1,\ldots,x_a\in\F_p}
    \prod_{i=1}^{a}x_i^s
    \left(
        1-(x_1+\cdots+x_a-1)^{p-1}
    \right).
\]
The contribution from the first term vanishes, since
\[
    \sum_{x\in\F_p}x^s=0
\]
for \(0<s<p-1\).

In the multinomial expansion of the remaining term, let \(k_0\) denote the
exponent of \(-1\), and let \(k_i\) denote the exponent of \(x_i\). A term
survives the summation over \(x_i\) only if
\[
    p-1\mid s+k_i.
\]
Since
\[
    0<s+k_i<2(p-1),
\]
this forces
\[
    s+k_i=p-1,
    \qquad
    k_i=K
\]
for every \(i\). Consequently,
\[
    k_0=p-1-aK,
\]
and every other exponent pattern vanishes. Using
\[
    \sum_{x\in\F_p}x^{p-1}=-1,
\]
we obtain
\[
    \rho_{\mathfrak p}\!\left(J_a(\chi)\right)
    =
    -(-1)^{p-1-aK}(-1)^a
    \frac{(p-1)!}{(p-1-aK)!(K!)^a}.
\]
Finally,
\[
    \frac{(p-1)!}{(p-1-aK)!}
    =
    \prod_{j=1}^{aK}(p-j)
    \equiv
    (-1)^{aK}(aK)!
    \pmod p.
\]
Combining the signs gives
\[
    \rho_{\mathfrak p}\!\left(J_a(\chi)\right)
    =
    (-1)^{a-1}\frac{(aK)!}{(K!)^a},
\]
which is equivalent to \eqref{eq:central}.

\section{Gauss--Jacobi Identities}\label{app:gauss-jacobi}

\begin{proof}[Proof of \Cref{prop:gauss-jacobi}]
Expanding \(G(\chi)^a\) and grouping the summands according to
\[
    s=x_1+\cdots+x_a
\]
gives
\[
    G(\chi)^a
    =
    \sum_{s\in\F_p}
    \exp\left(\frac{2\pi i s}{p}\right)
    \sum_{\substack{x_1,\ldots,x_a\in\F_p\\
                    x_1+\cdots+x_a=s}}
    \prod_{i=1}^a\chi(x_i).
\]
For \(s\neq0\), the change of variables \(x_i=sy_i\) shows that the inner
sum equals
\[
    \chi^a(s)J_a(\chi).
\]
The term \(s=0\) vanishes because \(\chi^a\) is nontrivial. Hence
\[
    G(\chi)^a
    =
    J_a(\chi)G(\chi^a),
\]
which proves the first identity. The second follows from the standard
Gauss-sum magnitude
\[
    |G(\psi)|=\sqrt p
\]
for every nontrivial multiplicative character \(\psi\)
\cite[Chapter~1]{BerndtEvansWilliams}.
\end{proof}

\section{The Stickelberger Factorization}\label{app:stickelberger}

For multiplicative characters \(\psi_1,\psi_2\) on \(\F_p^\times\), extended
by zero at the origin, write
\[
    J(\psi_1,\psi_2)
    =
    \sum_{x\in\F_p}\psi_1(x)\psi_2(1-x).
\]

\begin{lemma}\label{lem:multifold-jacobi-product}
For every \(1\le a<q\),
\[
    J_a(\chi)
    =
    \prod_{j=1}^{a-1}J(\chi^j,\chi),
\]
where the product is empty when \(a=1\).
\end{lemma}

\begin{proof}
For \(a=1\), both sides equal \(1\). Suppose \(2\le a<q\). Since
\(\chi^j\) and \(\chi^{j+1}\) are nontrivial for
\(1\le j\le a-1\), the two-character Gauss--Jacobi relation gives
\[
    J(\chi^j,\chi)
    =
    \frac{G(\chi^j)G(\chi)}{G(\chi^{j+1})}.
\]
Multiplying for \(j=1,\ldots,a-1\), the intermediate Gauss sums cancel:
\[
    \prod_{j=1}^{a-1}J(\chi^j,\chi)
    =
    \frac{G(\chi)^a}{G(\chi^a)}
    =
    J_a(\chi).
\]
\end{proof}

\begin{proposition}\label{prop:multifold-stickelberger}
For every \(1\le a<q\),
\[
    \bigl(J_a(\chi)\bigr)
    =
    \prod_{t\in(\Z/q\Z)^\times}
    \mathfrak p_t^{\left\lfloor at/q\right\rfloor}.
\]
\end{proposition}

\begin{proof}
Translated to the character and inverse-action conventions used here, the
standard prime-ideal factorization of a two-character Jacobi sum gives
\[
    \bigl(J(\chi^j,\chi^k)\bigr)
    =
    \prod_{t\in(\Z/q\Z)^\times}
    \mathfrak p_t^{c_t(j,k)},
\]
where
\[
    c_t(j,k)
    =
    \left\lfloor\frac{(j+k)t}{q}\right\rfloor
    -
    \left\lfloor\frac{jt}{q}\right\rfloor
    -
    \left\lfloor\frac{kt}{q}\right\rfloor
\]
\cite[Proposition~4(ii), pp.~83--84]{Katre2000}. Katre uses the negative of
our two-character Jacobi sum, which does not change the generated principal
ideal.

By \Cref{lem:multifold-jacobi-product},
\[
    \bigl(J_a(\chi)\bigr)
    =
    \prod_{j=1}^{a-1}\bigl(J(\chi^j,\chi)\bigr).
\]
Hence the exponent of \(\mathfrak p_t\) is
\begin{align*}
    \sum_{j=1}^{a-1}c_t(j,1)
    &=
    \sum_{j=1}^{a-1}
    \left(
        \left\lfloor\frac{(j+1)t}{q}\right\rfloor
        -
        \left\lfloor\frac{jt}{q}\right\rfloor
        -
        \left\lfloor\frac{t}{q}\right\rfloor
    \right) \\
    &=
    \left\lfloor\frac{at}{q}\right\rfloor,
\end{align*}
where the last equality uses \(1\le t<q\) and telescoping. This proves the
factorization.
\end{proof}

\section{Proof of the Central Congruence Modulo \texorpdfstring{\(p^2\)}{p squared}}
\label{app:lifted-central}

This appendix proves \Cref{prop:lifted-central}. The proof first
specializes Young's congruence between Jacobi sums and multinomial
coefficients and then expands the resulting quotient modulo \(p^2\).

\subsection{The Specialization of Young's Congruence}

Young's theorem relates generalized Jacobi sums to quotients of multinomial
coefficients. We first state the part of the theorem that is needed here.

Let \(\omega\) denote the Teichm\"uller character of
\(\F_p^\times\). Young uses the convention
\[
    J_{\mathrm Y}(\psi_1,\ldots,\psi_s)
    :=
    -
    \sum_{\substack{x_1,\ldots,x_s\in\F_p\\
                    x_1+\cdots+x_s=1}}
    \prod_{j=1}^{s}\psi_j(x_j).
\]
Thus \(J_{\mathrm Y}\) is the negative of the Jacobi sum convention used in
this paper.

Let
\[
    \alpha_1,\ldots,\alpha_s\in[0,1)\cap\Q
\]
have denominators not divisible by \(p\), and suppose that
\[
    A_j=(p-1)\alpha_j
\]
is an integer for every \(j\). Let
\[
    \alpha=\alpha_1+\cdots+\alpha_s,
\]
and assume that \(0<\alpha<1\). The second level of
\cite[Theorem~2.2]{Young1995} gives
\begin{equation}\label{eq:young-level-two}
    \frac{
        \displaystyle
        \binom{(p^2-1)\alpha}
              {(p^2-1)\alpha_1,\ldots,(p^2-1)\alpha_s}
    }{
        \displaystyle
        \binom{(p-1)\alpha}
              {(p-1)\alpha_1,\ldots,(p-1)\alpha_s}
    }
    \equiv
    (-1)^s
    J_{\mathrm Y}
    \left(
        \omega^{-A_1},\ldots,\omega^{-A_s}
    \right)
    \pmod{p^2},
\end{equation}
provided that the Jacobi sum on the right is not divisible by the
distinguished prime above \(p\). More generally, Young's theorem gives a
stronger modulus that also depends on its valuation at that prime.

We now specialize \eqref{eq:young-level-two} to the Jacobi sum occurring in
this paper.

\begin{lemma}\label{lem:young-specialization}
For every \(2\le a<q\),
\[
    \frac{
        \displaystyle
        \binom{a(p+1)K}
              {\underbrace{(p+1)K,\ldots,(p+1)K}_{a\text{ times}}}
    }{
        \displaystyle
        \binom{aK}
              {\underbrace{K,\ldots,K}_{a\text{ times}}}
    }
    \equiv
    (-1)^{a-1}J_a^{(2)}
    \pmod{p^2}.
\]
The multinomial coefficient in the denominator is invertible modulo \(p^2\).
\end{lemma}

\begin{proof}
In \eqref{eq:young-level-two}, take
\[
    s=a,
    \qquad
    \alpha_1=\cdots=\alpha_a=\frac1q.
\]
Since \(q\mid p-1\),
\[
    A_1=\cdots=A_a
    =
    \frac{p-1}{q}
    =
    K.
\]
Moreover,
\[
    \alpha=\frac aq<1.
\]

At the first level,
\[
    (p-1)\alpha_j=K,
    \qquad
    (p-1)\alpha=aK.
\]
At the second level,
\[
    (p^2-1)\alpha_j
    =
    \frac{p^2-1}{q}
    =
    (p+1)K,
\]
and
\[
    (p^2-1)\alpha=a(p+1)K.
\]
Thus the quotient in \eqref{eq:young-level-two} is precisely the quotient
displayed in the lemma.

It remains to identify the character and verify the divisibility condition.
Under the reduction
\[
    \zeta_q\longmapsto\widetilde\eta,
\]
the character \(\chi\) agrees with \(\omega^{-K}\). Indeed, on the primitive
root \(g\), both characters take the value \(\widetilde\eta\), the unique
\(q\)-th root of unity modulo \(p^2\) reducing to
\[
    \eta=g^{-K}\pmod p.
\]

By the Stickelberger factorization,
\[
    v_{\mathfrak p}\!\left(J_a(\chi)\right)
    =
    \left\lfloor\frac aq\right\rfloor
    =
    0.
\]
Hence the Jacobi sum is not divisible by the distinguished prime, so
\eqref{eq:young-level-two} applies modulo \(p^2\).

Finally, Young's Jacobi sum has the opposite sign from ours:
\[
    J_{\mathrm Y}
    \left(
        \omega^{-K},\ldots,\omega^{-K}
    \right)
    =
    -J_a(\chi).
\]
Since \(s=a\), the right-hand side of
\eqref{eq:young-level-two} becomes
\[
    (-1)^a\bigl(-J_a(\chi)\bigr)
    =
    (-1)^{a-1}J_a(\chi).
\]
Reducing through \(\rho_{\mathfrak p^2}\) gives the claimed congruence.

Since \(aK<p\), neither \((aK)!\) nor \(K!\) is divisible by \(p\).
Therefore,
\[
    \binom{aK}{K,\ldots,K}
    =
    \frac{(aK)!}{(K!)^a}
\]
is invertible modulo \(p^2\).
\end{proof}

\subsection{The Block Expansion}

\begin{lemma}\label{lem:p-free-block-product}
For every \(0\le m<p\),
\[
    \prod_{\substack{1\le j\le(p+1)m\\p\nmid j}}j
    \equiv
    ((p-1)!)^m\,m!
    \left(1+mpH_m\right)
    \pmod{p^2}.
\]
\end{lemma}

\begin{proof}
The product consists of \(m\) complete blocks followed by one final block:
\[
    \prod_{\substack{1\le j\le(p+1)m\\p\nmid j}}j
    =
    \prod_{h=0}^{m-1}
    \prod_{j=1}^{p-1}(hp+j)
    \prod_{j=1}^{m}(mp+j).
\]

For a complete block,
\[
\begin{aligned}
    \prod_{j=1}^{p-1}(hp+j)
    &=
    (p-1)!
    \prod_{j=1}^{p-1}
    \left(1+\frac{hp}{j}\right) \\
    &\equiv
    (p-1)!
    \left(1+hpH_{p-1}\right)
    \pmod{p^2}.
\end{aligned}
\]
Pairing \(j\) with \(p-j\) gives
\[
    H_{p-1}=0\pmod p,
\]
so every complete block is congruent to \((p-1)!\) modulo \(p^2\).

For the final block,
\[
\begin{aligned}
    \prod_{j=1}^{m}(mp+j)
    &=
    m!
    \prod_{j=1}^{m}
    \left(1+\frac{mp}{j}\right) \\
    &\equiv
    m!\left(1+mpH_m\right)
    \pmod{p^2}.
\end{aligned}
\]
Multiplying the blocks proves the lemma.
\end{proof}

\subsection{Derivation of the Congruence}

\begin{proof}[Proof of Proposition~\ref{prop:lifted-central}]
The case \(a=1\) follows from \(J_1(\chi)=1\), so assume
\[
    2\le a<q.
\]

Separating the multiples of \(p\) in the two multinomial coefficients gives
\[
    \frac{
        \displaystyle
        \binom{a(p+1)K}
              {\underbrace{(p+1)K,\ldots,(p+1)K}_{a\text{ times}}}
    }{
        \displaystyle
        \binom{aK}
              {\underbrace{K,\ldots,K}_{a\text{ times}}}
    }
    =
    \frac{
        \displaystyle
        \prod_{\substack{1\le j\le a(p+1)K\\p\nmid j}}j
    }{
        \displaystyle
        \left(
            \prod_{\substack{1\le j\le(p+1)K\\p\nmid j}}j
        \right)^a
    }.
\]
Applying Lemma~\ref{lem:p-free-block-product} with \(m=aK\) and \(m=K\)
gives
\[
\begin{aligned}
    \frac{
        \displaystyle
        \binom{a(p+1)K}
              {(p+1)K,\ldots,(p+1)K}
    }{
        \displaystyle
        \binom{aK}{K,\ldots,K}
    }
    &\equiv
    \frac{(aK)!}{(K!)^a}
    \frac{1+aKpH_{aK}}
         {(1+KpH_K)^a} \\
    &\equiv
    \frac{(aK)!}{(K!)^a}
    \left(
        1+aKp(H_{aK}-H_K)
    \right)
    \pmod{p^2}.
\end{aligned}
\]

Since
\[
    K=\frac{p-1}{q},
\]
we have
\[
    Kp
    =
    \frac{p(p-1)}q
    \equiv
    -\frac pq
    \pmod{p^2}.
\]
Therefore,
\[
    \frac{
        \displaystyle
        \binom{a(p+1)K}
              {(p+1)K,\ldots,(p+1)K}
    }{
        \displaystyle
        \binom{aK}{K,\ldots,K}
    }
    \equiv
    \frac{(aK)!}{(K!)^a}
    \left(
        1-\frac{ap}{q}(H_{aK}-H_K)
    \right)
    \pmod{p^2}.
\]

Combining this with Lemma~\ref{lem:young-specialization} gives
\[
    (-1)^{a-1}J_a^{(2)}
    \equiv
    \frac{(aK)!}{(K!)^a}
    \left(
        1-\frac{ap}{q}(H_{aK}-H_K)
    \right)
    \pmod{p^2}.
\]
Since
\[
    \frac{ap}{q}(H_{aK}-H_K)
\]
is divisible by \(p\), its square is zero modulo \(p^2\), and hence
\[
    \left(
        1-\frac{ap}{q}(H_{aK}-H_K)
    \right)^{-1}
    \equiv
    1+\frac{ap}{q}(H_{aK}-H_K)
    \pmod{p^2}.
\]
Multiplying by this inverse and by \((K!)^a\) proves
\[
    (aK)!
    \equiv
    (-1)^{a-1}J_a^{(2)}(K!)^a
    \left(
        1+\frac{ap}{q}(H_{aK}-H_K)
    \right)
    \pmod{p^2}.
\]
\end{proof}

\clearpage

\end{document}